\documentclass[conference,10pt]{IEEEtran}
\usepackage[T1]{fontenc}% optional T1 font encoding

\ifCLASSINFOpdf
\else
\fi

\usepackage{graphicx}
\usepackage{algorithm}
\usepackage{epstopdf}
\usepackage{footnote}
\usepackage{footmisc}
\usepackage{subfigure}
\usepackage{cuted}
\graphicspath{{./Figures/}}

\usepackage[cmex10]{amsmath}
\usepackage{amsfonts,latexsym,amssymb,amsthm}
\usepackage{epsfig,graphics,color}
\usepackage{setspace}
\usepackage{multirow}
\usepackage{cite}
\usepackage{algorithm, algpseudocode}
\usepackage{mathrsfs}
\usepackage{bm}
\usepackage{enumitem}
\setlist[itemize]{noitemsep, topsep=0pt}
\usepackage{xcolor}

\newtheorem{theorem}{Theorem}
\newtheorem{lemma}{Lemma}

\usepackage{soul}

\begin{document}
	\setlength{\columnsep}{0.17 in}
	\title{Enabling Large-Scale Federated Learning over Wireless Edge Networks}

	\author{
		\IEEEauthorblockN{Thinh~Quang~Dinh\IEEEauthorrefmark{1},  Diep~N.~Nguyen\IEEEauthorrefmark{1}, Dinh Thai Hoang\IEEEauthorrefmark{1}, Pham Tran Vu\IEEEauthorrefmark{2}\IEEEauthorrefmark{3}, and Eryk Dutkiewicz\IEEEauthorrefmark{1}}
		\IEEEauthorblockA{\IEEEauthorrefmark{1}University of Technology Sydney}
		\IEEEauthorblockA{\IEEEauthorrefmark{2}Ho Chi Minh City University of Technology (HCMUT), Vietnam}
		\IEEEauthorblockA{\IEEEauthorrefmark{3}Vietnam National University Ho Chi Minh City, Vietnam
			\\ Email: \{Thinh.Dinh, Diep.Nguyen, Hoang.Dinh, Eryk.Dutkiewicz\}@uts.edu.au, \{ptvu\}@hcmut.edu.vn}
	}
	

	\markboth{Journal of \LaTeX\ Class Files,~Vol.~14, No.~8, August~2015}%
	{Dinh \MakeLowercase{\textit{et al.}}: In-Network Computation Architecture for Very Large Scale Federated Learning}

	%\IEEEpubid{0000--0000/00\$00.00~\copyright~2017 IEEE}


	\maketitle
	
	\begin{abstract}
	Major bottlenecks  {of} large-scale Federated Learning (FL)  {networks} are the  {high costs for} communication and computation.  {This is due to the fact that most of current FL frameworks only consider} a star network topology where  {all local trained models are aggregated} at a single server (e.g., a cloud server).  {This} causes  significant overhead at the server when the number of users are huge and local models'  {sizes} are large.  This paper proposes a novel edge network architecture which decentralizes the  {model aggregation} process at the server, thereby significantly reducing the  aggregation latency of the whole network. In this architecture, we propose a highly-effective in-network computation protocol consisting of two components. First, an in-network aggregation process is designed  so that  {the majority of} aggregation  computations can be offloaded from cloud  server to edge nodes. Second, 
				  a joint routing and resource allocation optimization problem  {is formulated to minimize}  the aggregation latency  {for the whole system at every learning round}.  The problem turns out to be  {NP-hard, and thus we propose} a  {polynomial time} routing algorithm   which can achieve near optimal performance with a theoretical bound.  Numerical results  {show} that  {our proposed framework can dramatically reduce} the network latency, up to $4.6$ times.  {Furthermore, this framework can significantly decrease} cloud's traffic and computing overhead  {by a factor of $K/M$, where $K$ is the number of users and $M$ is the number of edge nodes, in comparison with conventional baselines}.
	\end{abstract}
	\begin{IEEEkeywords}
		Mobile Edge Computing, Federated Learning, In-network Computation
	\end{IEEEkeywords}

	% \fi
	\IEEEpeerreviewmaketitle

	\section{Introduction} \label{sec:intro}
	
	The last decade has witnessed the adoption of machine learning (ML) and artificial intelligence (AI) as the core engines of intelligent systems \cite{SaadNetMag2020}. Under most ML-based frameworks, raw data are collected and trained at centralized cloud  {servers,} raising concerns in user privacy,  latency, and network overhead.  {Federated Learning (FL)} has recently emerged as a potential distributed learning solution to these issues \cite{McMahan2017}. Under FL, mobile users (MUs), instead of sharing their raw data with the server, can build and learn their local learning  {models. After that, they only need to} send these local model parameters to the centralized server \cite{LimCOMST2020}.  {By doing so, the MUs can iteratively} download  the new global model from  the server, update their local models using  their local training data, and then upload their  {new local trained models} to the server for the model aggregation. This  process is  repeated  until   {the global model converges or after a predefined number of learning rounds reaches}. 
	
	However, given its distributed setting, communication and computation  costs are the two major bottlenecks of FL \cite{McMahan2017,LimCOMST2020,TranINFOCOM2019}.  {In addition, due to a huge demands of advanced AI-based mobile applications, learning tasks are more and more complicated with very large data sizes.} For example, with a large model like Visual Geometry Group-16 (VGG-16),  each user needs to update about $500$ TB of data until the global model is converged \cite{SamekTut2020}. Since conventional FL models use star network topologies,  {during the} model aggregation step, the cloud generally needs to connect with a huge number of users.  {In such a case,} aggregation operations at the cloud incur (a) high transmission latency, (b) high traffic overhead and (c) high computational overhead in term of processing and memory resources. To overcome these challenges,  {edge computing (EC) has recently emerged as a great potential solution by ``moving" computing resources} closer to end users \cite{LimCOMST2020}. Since edge nodes possess both computation and communication capacities, edge  {networks} can  decentralize  the model aggregation  computations at the cloud server in very large scale FL networks.  {To that end, it is critical to} develop a distributed in-network aggregation functionality implemented at edge networks'  {components in order to address current challenges of FL}.

	In-network computation (INC) is   a   process   {of gathering, processing data at intermediate nodes then routing the processed data through a multi-hop network} \cite{FasoloWC2007}. INC has been well-studied for distributed data clusters such as MapReduce \cite{DeanMapReduce}, Pregel \cite{GrzegorzPregel} and
	DryadLINQ \cite{YuanDryadLINQ}. Three  basic  components of an in-network computation solution are: suitable networking protocols, effective aggregation functions, and  efficient   {methods for}  data representation \cite{FasoloWC2007}. The early work of Liu \emph{et al.} \cite{LiuICC2020} proposed to aggregate/average  users' models at an edge node that later sends these intermediate model parameters to the cloud server. However, in this work, users are assumed to connect  directly to a single edge node without any alternative  {paths}. In practice,  {due to dense deployment of edge networks} \cite{ChenJSAC2018}, {a} given MU can associate one or another or even with multiple nearby edge nodes. As a result, the problems of network routing and resource allocation {for the model aggregation in FL under EC} become more challenging.

	 {Given the above, this paper proposes} a novel edge network architecture aiming at minimizing the  {aggregation} latency of  {FL processes}. This architecture allows the cloud node  {to decentralize} its  {aggregation process} to the edge nodes. To accomplish that network functionality, we design  an in-network computation protocol which consists of two components: an in-network aggregation process and a network routing algorithm.  {Specifically, the in-network aggregation process guides on how packets are processed at edge nodes and cloud node and how the cloud decentralizes the model aggregation process of FL.}
 Then,  {we formulate the joint routing and resource allocation optimization problem aiming to}  minimize the network's aggregation latency.  {The problem turns out to be NP-hard.} We   {thus} propose an effective algorithm based on randomized rounding techniques, which provably
 achieves  {an approximation guarantee}. Finally, simulation results  {show}  that our  {proposed solutions} significantly reduce not only the network' aggregation latency but also  {the} cloud node's traffic and computing overhead. 
%
%	
%		The rest of the paper is organized as follows. The system
%		model is introduced in Section \ref{section:sys_model}. We then design the in-network aggregation process in Section \ref{sec:ina}. After that, we  formulate the network routing and resource allocation frameworks in Section \ref{section:problem_for}, and propose our solution in Section \ref{section:solution}.   We then present the numerical results in Section \ref{section:numerical} and final conclusion in Section \ref{section:conclusion}.

	\section{System Model} \label{section:sys_model}

	As illustrated in Fig. \ref{fig::mqms},  let's consider a set of $K$  mobile users {MUs}, denoted by
	$\mathcal{K} =
	\{1,\cdots,K\}$ with local datasets $\mathcal{D}_k  = \{ \mathbf{x}_i \in \mathbb{R}^d, y_i \}_{i=1}^{n_k}$ with $n_k$ data points. Let $\mathcal{M} = \{0, \cdots ,M\}$ denote the set of edge nodes (ENs). They  {can be co-located with} small cell base stations which  {have} communications and computing capacities \cite{PoularakisTON2020}. These ENs are connected with a macro base stations, equipped with  a cloud server, denoted as EN $0$. Each user can be associated with one or more ENs. 
	\begin{figure}[t]
		\centering
		{\includegraphics[width=.95\linewidth]{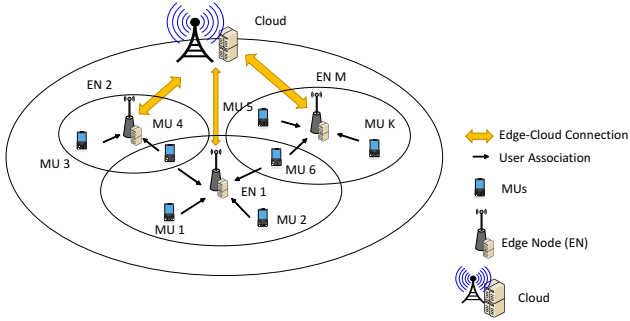}}
		\caption{ FL-enabled Edge Computing Network Architecture. \label{fig::mqms} }
	\end{figure}	
		\begin{figure}[h]
			\centering
			\subfigure[]{\includegraphics[width=0.18\linewidth]{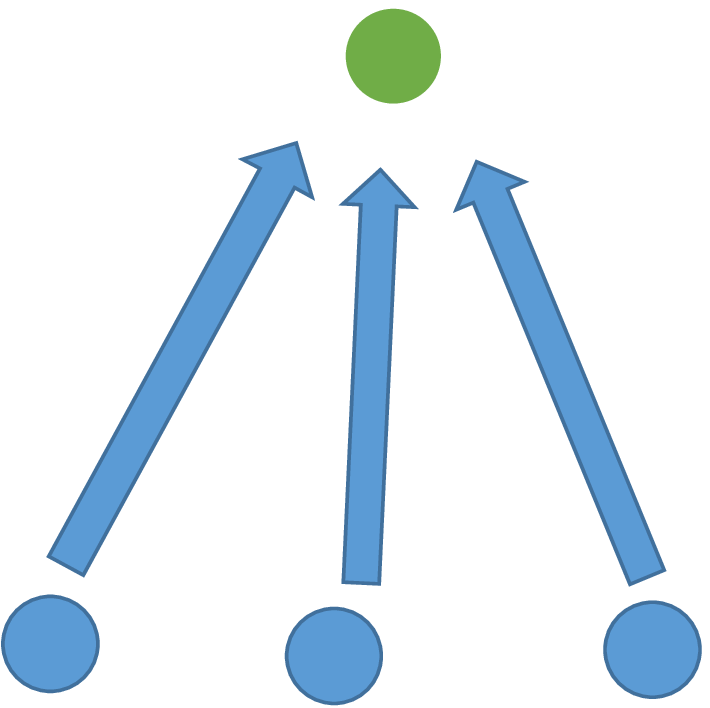}}\hspace{1.5cm}
			\subfigure[]{\includegraphics[width=0.43\linewidth]{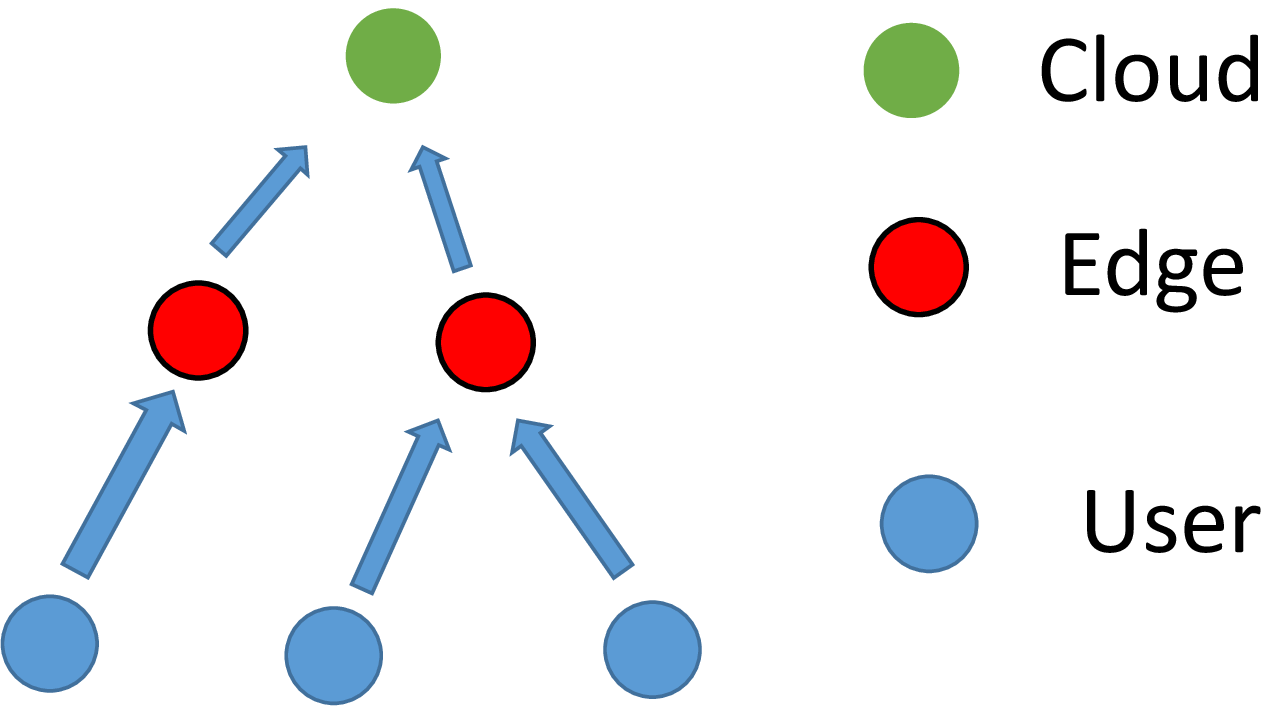}}
			\caption{The logical view of (a) conventional network model and (b) multi-tier edge network model.   \label{fig::logic_net} }
		\end{figure}
		
	\subsection{Federated Learning}
	 To construct the shared global model, the goal is to find the model  {parameters} $\mathbf{w} \in \mathbb{R}^d$ which  {minimize} the following global loss function in a distributed manner:
		\begin{align}
		\min_{\mathbf{w} \in \mathbb{R}^d} \Bigg\{ P(\mathbf{w}) = \frac{1}{n} \sum_{i=1}^{n} l_i (\mathbf{x}_i^T\mathbf{w})  + \xi r(\mathbf{w}) \Bigg\},  \label{eqn:ML_primal_prob}
		\end{align}
		where $n = \sum_{i=1}^{K} n_k $, $\xi$ is the regularizing parameter, $r(\mathbf{w})$ is a deterministic penalty function and $l_i$ is the loss function at data sample $i$ \cite{Ma2017}. Here, we also use notation $\bm{\psi}$ for the global model. 
	
		 {To solve (\ref{eqn:ML_primal_prob}), a Federated Learning framework introduced in \cite{McMahan2017} is performed as following.} At each iteration $t$, the cloud broadcasts the global model $\bm{\psi}^t$ to  {all the} MUs. Based on the latest global model, each MU 
		{learns}  {its} local parameters $\mathbf{w}_k^t$ according  {to the} Stochastic Gradient Descent update
		rule aiming  at minimizing the objective function $P(\mathbf{w})$ by only using local information and the parameter value
		in $\bm{\psi}^t$ \cite{McMahan2017}:
		\begin{align}
		\mathbf{w}_k^t = \bm{\psi}^{t} - \eta (\nabla l_i(\bm{\psi}^{t} 
		) + \nabla r(\bm{\psi}^{t})). \label{eqn:SGD}
		\end{align}
	The resulting local model updates are forwarded to the cloud for computing the new global model as follows:
		\begin{align}
		\bm{\psi}^{t+1}  = 
		\frac{1}{n} \sum_{k=1}^{K} n_k \mathbf{w}_k^t. \label{eqn:global_model}
		\end{align}

	We summarize the procedures of  {the FL framework} as follows:
	\begin{itemize}
				\item[1.] Global Model Broadcasting: The cloud broadcasts the latest global model $\bm{\psi}^{t}$ to the MUs.
		\item[2.] Local Model Updating: Each MU performs local training following (\ref{eqn:SGD}). 
		\item[3.] Global Model Aggregation: Local models are then sent back to the cloud.  The new value of global model is computed following (\ref{eqn:global_model}).
		\item[4.] Steps 1-3 are repeated until convergence.
	\end{itemize}

	\subsection{Communication Model}
	We then introduce the communication model for multi-user
	access.  For each FL iteration, the cloud node will select a set of users $\mathcal{K}^t$ at each iteration $t$ \footnote{The learner selection in FL can be based on the quality or significance of information or location learners \cite{LimCOMST2020}. Here, how to  {select best MUs at each learning round} is out of the scope of this paper.}.  All users consent about their models' structure, such as a specific neural network design. Hence, let $D$ denote the data size of model  {parameters, which is} fixed and identical for all users, 
	where  {$D$ is proportional} to the  {cardinality} of $\mathbf{w}$ \cite{SamekTut2020}. 
	\subsubsection{Global Model Broadcasting}
	Since the downlink communication capacity of the  {cloud node} is much larger than that of an edge node, all users  {will} listen to the  {cloud node} at the model broadcasting step. Let $W^d$ denote the downlink communication capacity of the cloud node.  The latency for broadcasting the global model is $T^{\mathrm{d}} = \frac{D}{W^d}$.

	\subsubsection{Global Model Aggregation}
		Let $a_{km} $ be the aggregation routing variable, where
		\begin{align}
	\nonumber	a_{km}  = \begin{cases}
		1 & \text{if MU $k$'s is associated with EN $m$}, \forall m \in \mathcal{M},\\
		0 & \text{otherwise.} 
		\end{cases}
		\end{align}
		  {Here, we assume that an MU is not  {allowed to transmit data directly to the} marco base station to reduce the uplink traffic overhead. However, MUs can listen to the downlink channel in network broadcast messages.} Let $\mathbf{a}_{m}   =  [a_{1m} , a_{2m} , \cdots, a_{Km}  ]^T$  {denote} the uplink association vector of edge node $m$. We let $\mathbf{A}  = \{a_{km} \}  \in \{0,1 \}^{K\times M}$  denote the uplink  {association matrix}, and $\tilde{\mathbf{a}}  = [\mathbf{a}_{0}^T,\mathbf{a}_{1}^T, \ldots , \mathbf{a}_{M}^T]^T$ denote the column vector corresponding to $\mathbf{A} $. Let $\mathcal{K}_m^t$ denote the set of user associated with edge node $m$,  {then we have} $\bigcup_m ~ \mathcal{K}_m^t = \mathcal{K}^t$, and $|\mathcal{K}^t| = \sum_{m}|\mathcal{K}_m^t|$,
		where $| \cdot |$  denotes the cardinality of a set. 	
		
		Let $r_{km}$ denote the uplink data rate between  {MU $k$ and edge node $m$.}	Let $\mathbf{r}_{m}   =  [r_{1m} , r_{2m} , \ldots, r_{Km}  ]^T$ denote the uplink bandwidth allocation vector corresponding to edge node $m$. 	We denote $\mathbf{R}  = r_{km} \in \mathbb{R}^{K\times M}$ as the uplink bandwidth allocation
		matrix.
	Let $B^{\mathrm{fr}}_m$, and $B^{\mathrm{bk}}_m$ denote the uplink fronthaul and backhaul capacity of edge node $m$.  Then, uplink communication latency between edge node $m$ and  {its associated users} is the longest latency of a given user:
	\begin{align}
	T^{\mathrm{u,\rm{fr}}}_m & = \max_{k \in \mathcal{K}_m} \Bigg\{D\frac{a_{km}}{r_{km}} \Bigg \}, ~\textrm{where}~ r_{km} \leq B^{\mathrm{fr}}_m, \forall m \in \mathcal{M} \setminus \{0\}. 
	\end{align}
	After edge nodes receive local models, each edge node can perform its aggregation computation, then send the aggregated result to the cloud node.  {Alternatively,} edge  nodes just forward received models to the cloud. Let $\gamma_m$ denote the transmission latency between edge node $m$ and the cloud node. Without in-network aggregation functionality, $\gamma_m$ is computed as followed
	\begin{align}
	\gamma_m = \frac{D\sum_{k \in \mathcal{K}_m }  a_{km}}{B^{\mathrm{bk}}_m}. \label{eqn:no_mec_edge_cloud_lat}
	\end{align}
	The uplink  {aggregation} latency of users associated with edge node $m$ is
	\begin{align}
	T^{\mathrm{u}}_m  & = T^{\mathrm{u,\rm{fr}}}_m + \gamma_m.
	\end{align}	 

	\section{In-Network Aggregation Design}	\label{sec:ina}	
	
	 {We now} introduce the in-network computation protocol where edge nodes support the cloud node for averaging users' local models. First, we design the user packet which plays a role of data representation in a in-network computation protocol. Let $\bm{\phi}_k^{t} = \{\phi_k^t[0],\bm{\phi}_k^t[1]  \}$ denote the local message of users $k$ at iteration $t$ such that
	\begin{align}
	\{ \phi_k^t[0],\bm{\phi}_k^t[1]  \} = \{ n_k, 	\mathbf{w}_k^t \}. 
	\end{align}
	\subsection{In-Network Aggregation Process}
	 {First, consider the following} in-network aggregation (INA) process at edge nodes and the cloud node  {that helps} decentralize the aggregation process at the cloud node. Let $\bm{\chi}^t_m$ denote the average local model of edge node $m$ such that
	\begin{align}
		\bm{\chi}^t_m  = \frac{1}{\sum_{k \in \mathcal{K}_m^t} \phi_k^t[0]}\sum_{k \in \mathcal{K}_m^t} \phi_k^t[0] \bm{\phi}_k^t[1]. \label{eqn:edge_model}
		\end{align}	
		Let $\varphi_m^t = \{\varphi_m^t[0],\bm{\varphi}_m^t[1] \}$ denote the message edge node $m$ sends to  {the cloud} node such that
		\begin{align}
		\{ \varphi_m^t[0],\bm{\varphi}_m^t[1]  \} = \Big \{ \sum_{k \in \mathcal{K}_m^t} \phi_k^t[0] , \bm{\chi}^t_m \Big \}.
		\end{align}
	 {To conserve the result of (\ref{eqn:global_model}),} the global model is computed as follows:
	\begin{align}
	\bm{\psi}^{t+1} &= \frac{ \sum_{m}\varphi_m^t[0]\bm{\varphi}_m^t[1]}{ \sum_{m}\varphi_m^t[0]} .  \label{eqn:agg_proc:final} 
	\end{align}
 \begin{theorem}
	The edge network architecture as well as the INA process reduce the traffic and computing overhead at the cloud node by a factor of $K/M$ in comparison with conventional star network topologies.
\end{theorem}
		\begin{proof}
			The proof is omitted here for brevity. 
		\end{proof}
	\subsection{Revised Latency Model}
	With the proposed INA process, let $\gamma_m'$ denote the transmission latency between an edge node $m$ and the cloud node. If there is no user associate with an edge node $m$, i.e., $\sum_{k}  a_{km} = 0$, $\gamma_m'$ is zero. Otherwise, since edge node $m$ only needs to send  {its} aggregated model,  {computed in (\ref{eqn:edge_model}),}  {to} the cloud node, $\gamma_m'$ is computed as follows
	\begin{align}
		\gamma_m' = \min \Bigg \{\frac{D}{B^{\mathrm{bk}}_m}, \frac{D\sum_{k \in \mathcal{K}_m }  a_{km}}{B^{\mathrm{bk}}_m} \Bigg \}.
	\end{align}

	\section{Network Routing and Resource Allocation Framework for FL} \label{section:problem_for}
	In this section,  we  aim to minimize the  total uplink aggregation latency by jointly  optimizing (a)  which edge node a user should send its local  {model} directly to and (b) the optimal data  rates for wireless connections between the users and the edge nodes. The total uplink aggregation latency is  computed as follows
	\begin{align}
	T^{\mathrm{u}}(\mathbf{A},\mathbf{R}) & =  \max_m \Bigg\{  T^{\mathrm{u,\rm{fr}}}_m  + \gamma_m' \Bigg\}  .
	\end{align}
	The aggregation latency-minimized routing framework is formulated as  followed
	\begin{subequations}
			\begin{align}
			\nonumber \mathscr P_1 : & \min_{\mathbf{A},\mathbf{R}} T^{\mathrm{u}}(\mathbf{A},\mathbf{R}),  \\
			\rm{s.t.}~	& \sum_{m=0}^{M}  a_{km} = 1, \forall k \in \mathcal{K}^t, \label{eqn:constraint:assignment}\\
			& \sum_{k \in \mathcal{K}_m} r_{km} \leq B^{\mathrm{fr}}_m, \forall m \in \mathcal{M} \setminus \{0\}, \label{eqn:constraint:edge_bw}\\ 
			& a_{km} \in \{0,1\},\label{eqn:constraint:assign_variable}\\
			& r_{km} \in[0,B^{\mathrm{fr}}_m] , \forall m \in \mathcal{M} \setminus \{0\}. \label{eqn:constraint:bw_variable}
			\end{align}
	\end{subequations}
	The constraints (\ref{eqn:constraint:assignment}) guarantee that a user can associate with only one edge node in one iteration. The constraints (\ref{eqn:constraint:edge_bw}) ensure that total users' data rates associated with each edge node must
	not exceed its bandwidth capacity. $\mathscr P_1 $ is  a  mixed-integer  nonlinear  programming, which  is  NP-hard.\footnote{ {The proof is omitted here for brevity.}}  We will propose a  {highly efficient} randomized rounding solution for practical implementation in the next section.

	\section{Randomized Rounding Based Solution} \label{section:solution}
	
	In this section, we present an approximation algorithm for the main
	problem that leverages a randomized rounding technique  {\cite{motwani1996randomized}}.  Firstly, $\mathscr P_1 $ is  transformed  to  an equivalent integer linear program (ILP). Then,  by  relaxing the integer constraints, $\mathscr P_1 $ becomes a linear programming which can be solved by  linear  solvers.  {We first observe that}:
	
		\begin{lemma}\label{lemma:2}
			Given any uplink  {association matrix} $\mathbf{A}$, with $ |\mathcal{K}_m| = \sum_{k \in \mathcal{K}_m}  a_{km} >0$, for problem $\mathscr P_1 $, at each edge node $m$, the uplink latency for users associated with edge node $m$ satisfies
			\begin{align}
			T^{\mathrm{u}}_m  & = \max_{k \in \mathcal{K}_m} \Bigg\{D\frac{a_{km}}{r_{km}} \Bigg \} + \min \Bigg \{\frac{D}{B^{\mathrm{bk}}_m}, \frac{D\sum_{k \in \mathcal{K}_m }  a_{km}}{B^{\mathrm{bk}}_m} \Bigg \} \nonumber \\
			&  \geq \frac{D |\mathcal{K}_m|}{B^{\mathrm{fr}}_m}  + \min \Bigg \{\frac{D}{B^{\mathrm{bk}}_m}, \frac{D|\mathcal{K}_m|}{B^{\mathrm{bk}}_m} \Bigg \}.
			\end{align}
			The equality happens when $r_{1m} = \ldots = r_{|\mathcal{K}_m|m} = \frac{B^{\mathrm{fr}}_m}{|\mathcal{K}_m|} $.
		\end{lemma}
		\begin{proof}
			The proof is omitted here for brevity. 
		\end{proof}	
		Following Lemma \ref{lemma:2}, the network operator hence only needs to  optimize the uplink  {association matrix} while the uplink data rates for users associated with edge nodes will be allocated  {in a} fairness manner. If $|\mathcal{K}_m| = 0$, we arbitrarily set the  {values} of $\mathbf{r}_{m}$ and $T^{\mathrm{u}}_m$  {to be} $0$. As a result, $\mathscr P_1 $ is reduced to
		\begin{align}
		\nonumber \mathscr P_2 : & \min_{\mathbf{A}}  \max_m \Bigg\{ D\frac{ \sum_{k \in \mathcal{K}_m}  a_{km}}{B^{\mathrm{fr}}_m} \\
	\nonumber	&\hspace{1.5cm}  + \min \Bigg \{\frac{D}{B^{\mathrm{bk}}_m}, \frac{D\sum_{k \in \mathcal{K}_m }  a_{km}}{B^{\mathrm{bk}}_m} \Bigg\} \Bigg\},\\
		\rm{s.t.}~	& (\ref{eqn:constraint:assignment})~ \textrm{and} ~(\ref{eqn:constraint:assign_variable}),
		\end{align}
		 {where the the optimal solution in $\mathscr P_1$ can be computed from optimal solution in $\mathscr P_2 $ as the following lemma:}

\begin{lemma}
	 {Let $\mathbf{A}^{**}$ denote the optimal solution in $\mathscr P_2 $. Following Lemma \ref{lemma:2}, the optimal solution in $\mathscr P_1 $, $\{\mathbf{A}^*,\mathbf{R}^*\}$ is computed as follows}
	\begin{align}
	\nonumber	& \mathbf{A}^*  = \mathbf{A}^{**}, \\
		& r_{1m}^* = \ldots = r_{|\mathcal{K}_m|m}^* = \frac{B^{\mathrm{fr}}_m}{|\mathcal{K}_m|}.
	\end{align}
\end{lemma}	
			\begin{proof}
				The proof is omitted here for brevity. 
			\end{proof}	
	
	The proposed algorithm is described in detail
	below and summarized in Algorithm 1. First, we introduce auxiliary variables $y$ and $\bm{\gamma}= \{\gamma_1,\ldots,\gamma_M \}$ into $\mathscr P_2 $ such that
	\begin{align}
	y & \geq  \max_m \Bigg\{D \frac{ \sum_{k \in \mathcal{K}_m}  a_{km}}{B^{\mathrm{fr}}_m} + \gamma_m \Bigg\},\\
	\gamma_m & \leq \min \Bigg \{\frac{D}{B^{\mathrm{bk}}_m}, \frac{D\sum_{k \in \mathcal{K}_m }  a_{km}}{B^{\mathrm{bk}}_m} \Bigg \}, \forall m \in \mathcal{M} \setminus \{0\},
	\end{align}
	Problem $\mathscr P_2 $ is  {then equivalently} transformed to
	\begin{subequations}
	\begin{align}
	\nonumber \mathscr P_3 : & \min_{y,\bm{\gamma},\mathbf{A}} y, \\
	\mathrm{s.t.}~& y \geq D\frac{ \sum_{k \in \mathcal{K}_m}  a_{km}}{B^{\mathrm{fr}}_m} + \gamma_m, \forall m \in \mathcal{M} \setminus \{0\}, \label{eqn:aux_cons:inc:edge_lat}\\
	& \gamma_m \leq \frac{D}{B^{\mathrm{bk}}_m}, \forall m \in \mathcal{M} \setminus \{0\}, \label{eqn:aux_cons:inc:ina} \\
	& \gamma_m \leq \frac{D\sum_{k \in \mathcal{K}_m }  a_{km}}{B^{\mathrm{bk}}_m}\label{eqn:aux_cons:inc:non_ina}, \forall m \in \mathcal{M} \setminus \{0\}, \\
	& (\ref{eqn:constraint:assignment})~ \textrm{and} ~(\ref{eqn:constraint:assign_variable}). \nonumber
	\end{align}
	\end{subequations}
	The Algorithm \ref{algorithm:heu_per_plan} starts by solving the Linear Relaxation (LR) of $\mathscr P_3 $. Specifically, it relaxes the variables
	$a_{km}$ to be fractional, rather than integer. The
	Linear Relaxation of $\mathscr P_3 $ can be expressed as follows:
			\begin{align}
			\nonumber \mathscr P_4 : & \min_{y,\bm{\gamma},\mathbf{A}} y,  \\
			\mathrm{s.t.}~& (\ref{eqn:aux_cons:inc:edge_lat})-(\ref{eqn:aux_cons:inc:non_ina}), (\ref{eqn:constraint:assignment}),~\textrm{and}~ a_{km} \in [0,1].
			\end{align}
	
	Let $z^{\dagger}= [{\tilde{\mathbf{a}}}^{\dagger},y^{\dagger}, \gamma_1^{\dagger}, \ldots,\gamma_M^{\dagger}]$ denote the optimal solution of $\mathscr P_4 $.
	First,  {vector ${\tilde{\mathbf{a}}}^{\dagger}$ should be transformed}  to  {an} equivalent fractional matrix ${\mathbf{A}}^{\dagger}$,
	whose elements are in $[0, 1]$ by a ``reshape'' operation. The
	term ``reshape'' means to change the size of a vector or a matrix
	while its number of elements is unchanged.  {${\mathbf{A}}^{\dagger}$} is the optimal solution to $\mathscr P_2 $, if all components
	of ${\mathbf{A}}^{\dagger}$ are binary. Otherwise,
	to  {obtain binary matrix ${\mathbf{A}}^{(\rm{Alg})}$}, for each row of ${\mathbf{A}}^{\dagger}$,
	we perform a randomization by setting the element $a_{km}$ to 1 with probability ${a_{km}}^{\dagger}$. The decision is done in an exclusive manner for satisfying constraints (\ref{eqn:constraint:assignment}). It means that for each row $k$, only one element of the row is one, the rest  {are zeros}. The random decision is made independently for all $k$. By doing
	this procedure,  {the matrix ${\mathbf{A}}^{(\rm{Alg})}$ is achieved}. Then, ${{r_{1m}} }^{(\rm{Alg})} = \ldots = {r_{|\mathcal{K}_m|m} }^{(\rm{Alg})} = \frac{B^{\mathrm{fr}}_m}{|\mathcal{K}_m|}, \forall m$. The complexity of this algorithm is $O(\nu^{3.5}\Omega^2)$, where $\nu = KM+ M+1$.
	
	\begin{algorithm}    [t]                % enter the algorithm environment
		\caption{Randomized Routing Algorithm for Low Latency Federated Learning}          % give the algorithm a caption
		\label{algorithm:heu_per_plan}                           % and a label for \ref{} commands later in the document
		{\footnotesize		
			\begin{algorithmic}[1]                    % enter the algorithmic environment
				\Require $D$, $B^{\mathrm{fr}}_m$, $B^{\mathrm{bk}}_m$, and $\mathcal{K}^t$.
				\Ensure ${\mathbf{A} }^{(\rm{Alg})}$, ${\mathbf{R} }^{(\rm{Alg})}$
				\State Solve $\mathscr P_4 $  to achieve ${\mathbf{A}}^{\dagger}$.
				\If  {${\mathbf{A}}^{\dagger}$ is binary}
				\State ${\mathbf{A} }^{(\rm{Alg})} = {\mathbf{A}}^{\dagger}$
				\Else
				\For {$k = 1$ to $k= K$}
				\State {${a_{km}}^{(\rm{Alg})} = 1$ with probability ${a_{km}}^{\dagger}$} with exclusive manner based on constraints (\ref{eqn:constraint:assignment})
				\EndFor
				\EndIf						
				
				\State {Then,
					\begin{align}
					\nonumber 	&	{{r_{1m}} }^{(\rm{Alg})} = \ldots = {r_{|\mathcal{K}_m|m} }^{(\rm{Alg})} = \frac{B^{\mathrm{fr}}_m}{|\mathcal{K}_m|}, \forall m ~\textrm{such that}~|\mathcal{K}_m| \neq 0.
					\end{align}
				}
			\end{algorithmic}}
		\end{algorithm}
		\begin{theorem}
			The aggregation latency returned by Algorithm \ref{algorithm:heu_per_plan} is at most $\frac{2 \ln K}{y^{\dagger}} + 3$ times higher  {than that of} the optimal with
			high probability,  {where $K$ is the number of MUs and $y^{\dagger}$ is the lower bound of the aggregation latency which can be obtained in polynomial time.}
		\end{theorem}
		\begin{proof}
			The proof is omitted here for brevity. 
		\end{proof}

		\section{Numerical Results} \label{section:numerical}
				\begin{figure}[t]
					\centering
					{\includegraphics[width=1\linewidth]{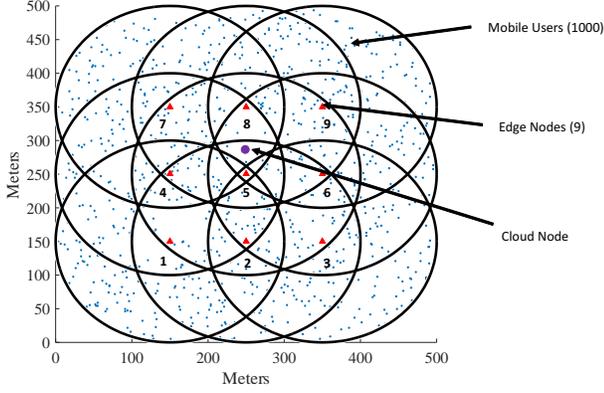}}
					\caption{The network setup. \label{fig::network_setup} }
				\end{figure}
		In this section, {simulations}  {are conducted} to show the
		performance of the proposed algorithm. We consider a
		similar setup as in \cite{PoularakisTON2020}, depicted in  {Fig.} \ref{fig::network_setup}.
		Here, $M = 9$ edge nodes are regularly deployed  {in} a
		grid network inside a $500 \times 500 ~\rm{m}^2$ area. $K = 1000$ mobile
		users are distributed uniformly at random over the EN coverage
		regions (each of $150$m radius). In our simulation, without loss of generality, all $K$ users' models are aggregated in one learning iteration. The cloud node's coverage contains all mobile users. For each edge node $m$,
		we set the the uplink
		fronthaul capacity to $B^{\mathrm{fr}}_m = 1$Gbps, the backhaul capacity to $B^{\mathrm{bk}}_m = 1$Gbps. These settings are inspired by Wifi IEEE 802.11ac standards \cite{Wifiax}, and data centers interconnection using optical fibers \cite{Chengthesis2019}. We also set the cloud downlink capacities $W^d = 2$Gbps. These values  {may be changed} during the evaluations. For model aggregation, by default, we investigate our system using ResNet152's model size, i.e., $D = 232$ MB \cite{KerasPretrain}. In later simulations, we also investigate our system with different model sizes.

		\subsection{Algorithm Comparison - Latency Reduction}
				\begin{figure}[t]
					\centering
					{\includegraphics[width=0.9\linewidth]{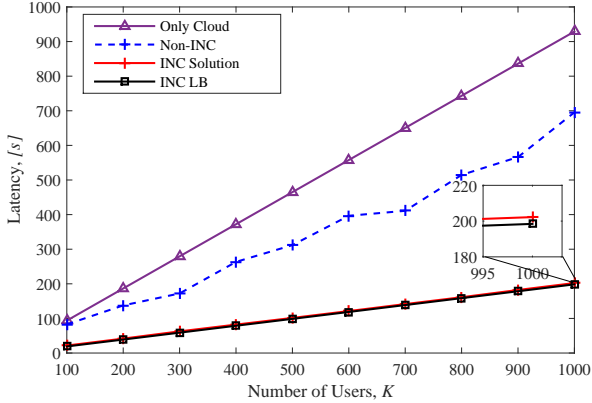}}
					\caption{Algorithm comparison with respect to different number of users. \label{fig::alg_com_change_K} }
				\end{figure}
		Fig. \ref{fig::alg_com_change_K} compares the aggregation latency of different algorithms
		versus the number of users $K$ in one learning iteration. The proposed INC protocol is
		compared with three other baseline methods, namely:
		\begin{itemize}
			\item[1.] \textit{Only Cloud}: $K$ users send their models to the cloud node via its hypothetical uplink wireless channel with $W^u = 2$Gbps.
			\item[2.] \textit{INC Solution}: $K$ users can associate with the cloud node and edge nodes with INC protocol.  The network routing problem $\mathscr P_3 $ is solved by using Algorithm \ref{algorithm:heu_per_plan}.
						\item[3.] \textit{Non-INC}: In this scenario, without the proposed INA process, $K$ users are associated with their nearest edge nodes regardless of their capacities. The latency between edge node and cloud node is computed by following (\ref{eqn:no_mec_edge_cloud_lat}).
			\item[4.] \textit{INC LB}: In this scenario, we use Linear
			Relaxation to solve $\mathscr P_3 $. This scenario will provide the lower bound of network latency if the proposed INC protocol is considered.
		\end{itemize}
		As can be  {observed in} Fig. \ref{fig::alg_com_change_K}, our proposed algorithm can achieve near optimal performance. When $K = 1000$, 
		 {the latency obtained by the proposed solution is approximately $1.9\%$ higher than that of the \textit{INC LB}}. It implies that  our proposed solution can achieve the performance almost the same as that of the lower bound solution. \textit{Only Cloud} has the worst performance. For example, when $K = 1000$, the aggregation latency of \textit{Only Cloud} is $928.9$s which is  $133\%$ higher than that of the second worst one, \textit{Non-INC}, $695.3$s. We also observe that when $K = 1000$, \textit{Only Cloud} and \textit{Non-INC} are $459\%$ and $343\%$ higher than  {that of our} proposed solution, i.e., \textit{INC Solution}. Last but not least, the  {gaps} between our proposed algorithm and   \textit{Only Cloud} and \textit{Non-INC} enlarge as the number of users $K$ increases.  {This clearly shows} that our  {proposed solution} is  {significantly beneficial} for very large scale federated learning networks.

		\begin{figure}[t]
			\centering
			{\includegraphics[width=0.9\linewidth]{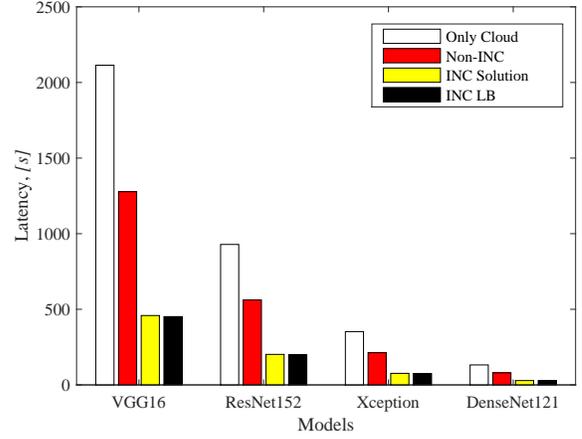}}
			\caption{Algorithm comparison with respect to different models. \label{fig::alg_com_change_D} }
		\end{figure}

		 {In Fig. \ref{fig::alg_com_change_D}, consider \textit{Only Cloud}, \textit{Non-INC} and  our \textit{INC solution}, we evaluate the aggregation latency in different models in one learning iteration. They are VGG16, ResNet152, Xception and DenseNet121 whose model sizes are $528$ MB, $232$ MB, $88$ MB and $33$ MB, respectively \cite{KerasPretrain}. Here, we choose the default setting  {with} $K=1000$. We observe that with different models, the aggregation latency of the proposed solution, \textit{INC solution}, is significantly lower  {than those of the} \textit{Only Cloud} and \textit{Non-INC}. For example, with VGG16, the aggregation latency of \textit{INC solution} is $4.6$ times and $2.7$ times lower than  {those} of \textit{Only Cloud} and \textit{Non-INC}, respectively. 
		
		\subsection{Traffic and Computation Reduction at the Cloud Node}
		
		\begin{figure}[t]
			\centering
			{\includegraphics[width=0.85\linewidth]{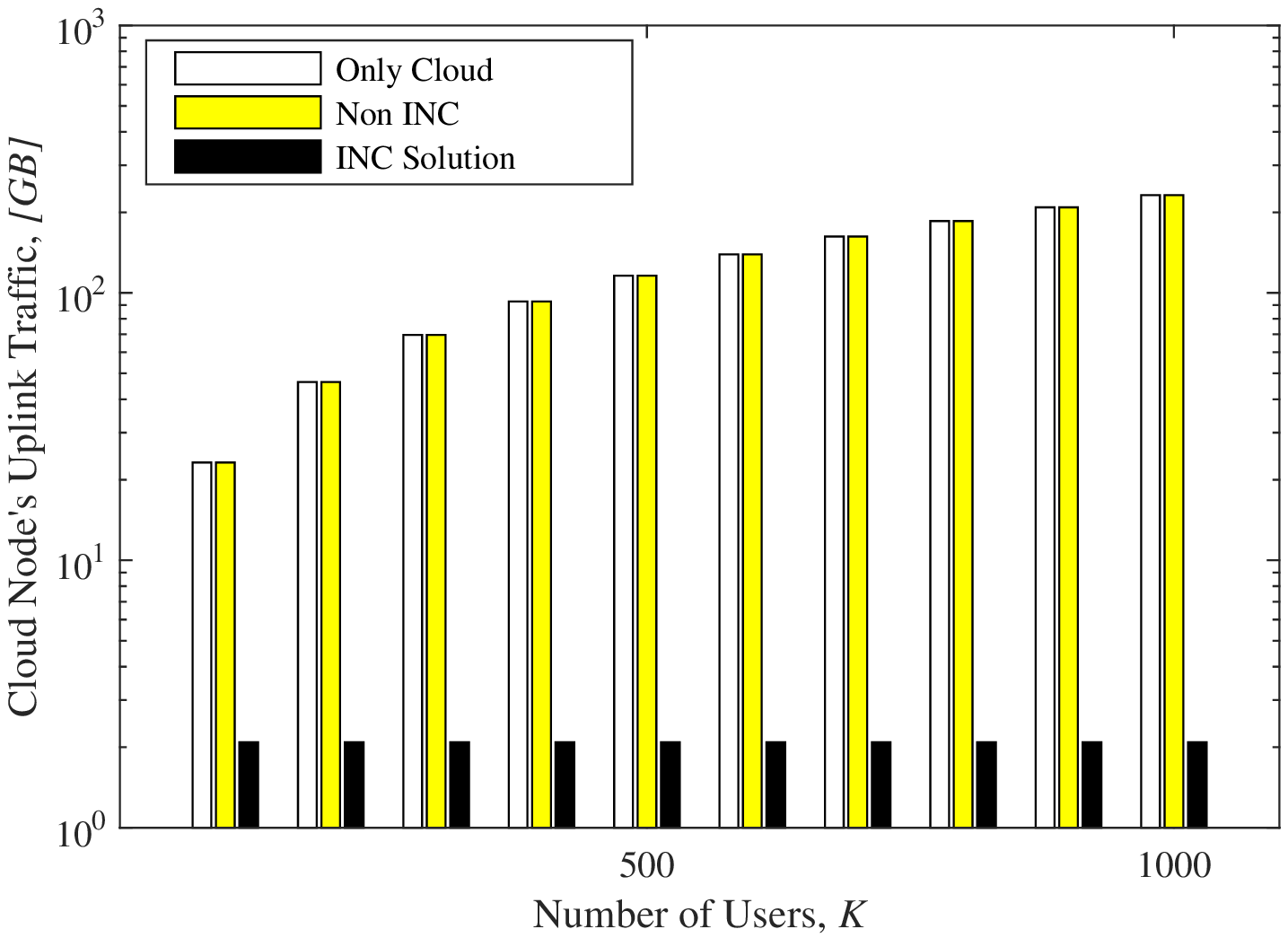}}
			\caption{Cloud node's uplink traffic and computing load. \label{fig::network_load} }
		\end{figure}
		
				\begin{figure}[t]
					\centering
					{\includegraphics[width=0.85\linewidth]{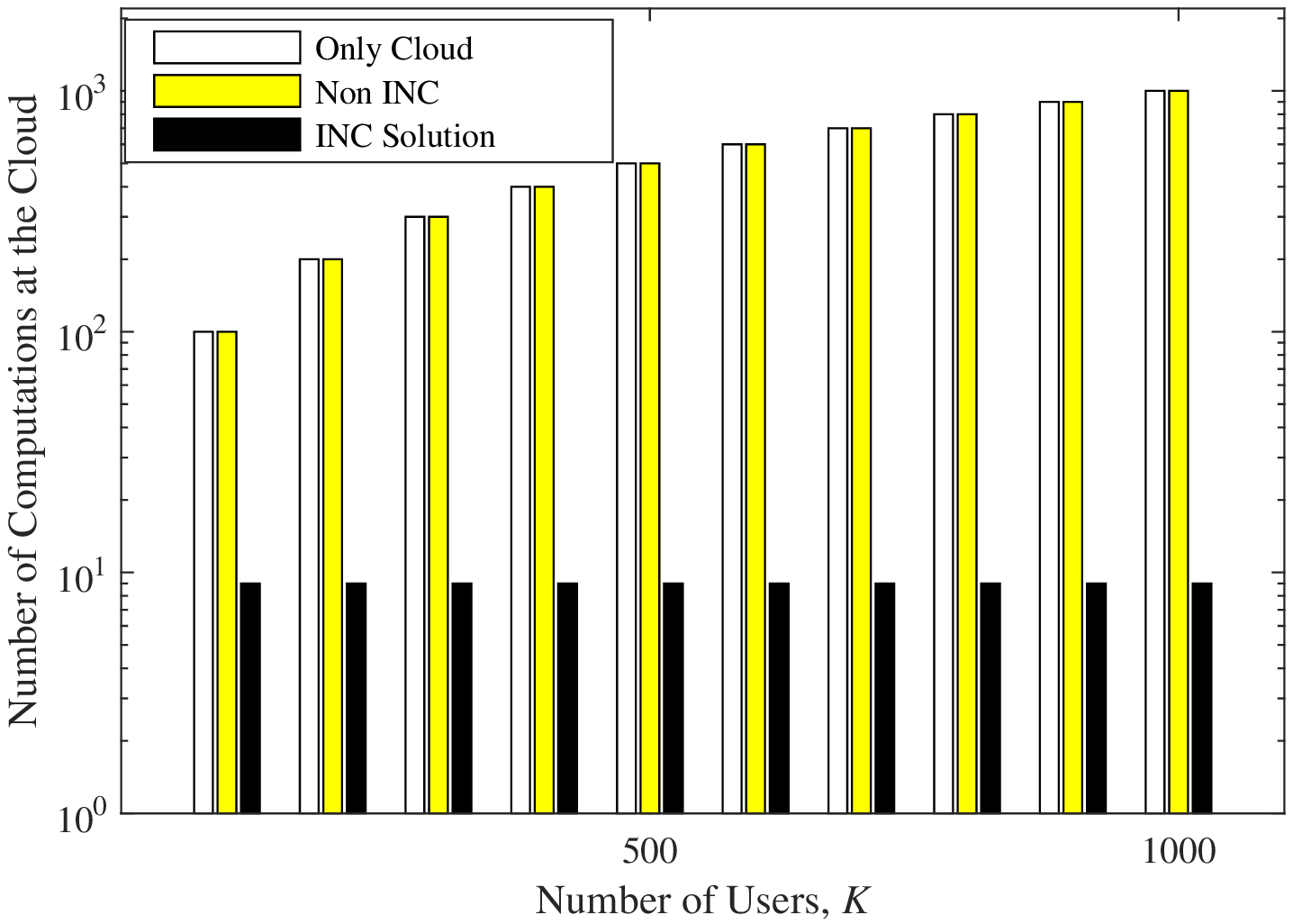}}
					\caption{Cloud node's uplink traffic and computing load. \label{fig::computing_load} }
				\end{figure}
		In this part, using ResNet152's model setting, we investigate the uplink traffic and the number of models needed to be aggregated at the cloud node in one learning iteration. We compare three schemes: \textit{Only Cloud}, \textit{Non-INC} and  our \textit{INC solution}. The number of models needed to be aggregated at the cloud node is proportional to the number of computations here. In Fig. \ref{fig::network_load} and Fig. \ref{fig::computing_load}, the uplink traffic and the number of computations of \textit{Non-INC} at the cloud node  {are} equal to those of \textit{Only Cloud}. It is because all models need to be sent to the cloud before  {being} aggregated and edge nodes only  {forward} the models from users to the cloud without the proposed INA process. Meanwhile, with \textit{INC solution}, the two metrics are significantly reduced by remaining unchanged at low values. The reason is that the cloud only collects aggregated models from edge nodes which are fixed. For example, when $K = 1000$, the traffic is $2.32$GB for our scheme and $232$GB for the other two.  As a result, our scheme can keep the traffic and computing load at the cloud very low even with a large number of users.

		\section{Conclusion} \label{section:conclusion}
		In this paper, we propose a  {novel} edge network architecture aiming at minimizing the aggregation latency of  {FL processes}. This architecture is able to   decentralize the model aggregation process of cloud node to edge nodes.  To achieve that network functionality, we design an in-network computation protocol  consisting of an in-network aggregation process and a network routing algorithm. The in-network aggregation  {process is to enhance learning processes through leveraging computations at the edges and cloud.} We also formulate  a joint routing and resource allocation optimization problem to minimize the network's aggregation latency.  {As} the optimization problem is NP-Hard, we propose  {a highly-effective solution} based on randomized rounding with provable performance guarantee.  Our simulation results
		 {show} that the proposed algorithm  {can achieve} near
		optimal network latency and outperform some  {other baseline} schemes such as \textit{Only Cloud}, \textit{Non-INC}. We also show that  {the INC} protocol  {can help} the cloud node significantly decrease not only its network's aggregation latency but also its traffic load and computing load.
		
		\section{Acknowledgment}
		This work was supported in part by the Joint Technology and Innovation Research Centre, a partnership between University of Technology Sydney and Ho Chi Minh City University of Technology (HCMUT) - VNU HCM.
		\appendices


		\ifCLASSOPTIONcaptionsoff
		\newpage
		\fi

		\bibliographystyle{IEEEtran}
		\bibliography{IEEEabrv,references}

	\end{document}